\documentclass[12pt, draftclsnofoot, journal, letter, onecolumn]{IEEEtran}
\ifCLASSINFOpdf
\else
\fi

\usepackage{graphicx}
\usepackage{ifpdf}
\usepackage{amsmath}
\usepackage{amsfonts}
\usepackage{epsfig}
\usepackage{amssymb}
\usepackage{graphicx}

\newtheorem{proposition}{{\textbf{Proposition}}}
\newtheorem{lemma}{{ \textbf{Lemma}}}

\newtheorem{definition}{{ \textbf{Definition}}}

\begin{document}
%
\title{Truthful Spectrum Auction for Efficient Anti-Jamming in Cognitive Radio Networks}
%
%
%

\author{Mohammad~Aghababaie~Alavijeh,
        Behrouz~Maham,~\IEEEmembership{Senior~Member,~IEEE,}
        Zhu~Han,~\IEEEmembership{Fellow,~IEEE,}
        and~Walid~Saad,~\IEEEmembership{Senior~Member,~IEEE}
\thanks{Preliminary version of a portion of this work is appeared in
\emph{Proc. IEEE International Conference on Communications
(ICC'13)}.} \thanks{ Mohammad Aghababaie Alavijeh
is with the School of ECE, College of Engineering, University of
Tehran, Iran. Email: aghababaie@ut.ac.ir. Behrouz Maham is with the
School of Engineering, Nazarbayev University, Astana, Kazakhstan.
Email: behrouz.maham@nu.edu.kz. Zhu Han is with the Department of
ECE, University of Houston, Houston, TX 77004, USA. E-mail:
zhan2@uh.edu. Walid Saad is with the Department of ECE, Virginia
Tech, Blacksburg, VA 24060, USA. E-mail: walids@vt.edu.
}
}
\maketitle

\begin{abstract}
One significant challenge in cognitive radio networks is to design a
framework in which the selfish secondary users are obliged to
interact with each other truthfully. Moreover, due to the
vulnerability of these networks against jamming attacks, designing
anti-jamming defense mechanisms is equally important.
In this paper, we propose a truthful mechanism, robust against the
jamming, for a dynamic stochastic cognitive radio network consisting
of several selfish secondary users and a malicious user. In this
model, each secondary user participates in an auction and wish to
use the unjammed spectrum, and the malicious user aims at jamming a
channel by corrupting the communication link. A truthful auction
mechanism is designed among the secondary users. Furthermore, a
zero-sum game is formulated between the set of secondary users and
the malicious user. This joint problem is then cast as a randomized
two-level auctions in which the first auction allocates the vacant
channels, and then the second one assigns the remaining unallocated
channels. We have also changed this solution to a trustful
distributed scheme. Simulation results show that the distributed
algorithm can achieve a performance that is close to the centralized
algorithm, without the added overhead and complexity.
\end{abstract}


%
\IEEEpeerreviewmaketitle

\section{Introduction}
Spectrum scarcity has been a major problem for the existing wireless
networks which motivated researchers to investigate new intelligent
paradigm to manage available spectrum. Cognitive radio (CR) has thus
emerged as a promising approach to improve spectral efficiency in
wireless networks. In CR networks, secondary users (SUs) may
cognitively access unused spectrum that is not currently occupied by
licensed users, namely primary users (PUs) under the condition that
the PUs' transmission will not be interfered~\cite{QZhao}.

Spectrum management in CR networks has been considered in many
recent works such as ~\cite{Akyildiz} and \cite{Management2013} (and
references therein). One important technique that enables
CR-oriented spectrum allocation is to consider spectrum auction
among SUs that seek to idle channels~\cite{Designing Truthful
Spectrum Auctions}. Auction theory, which is rooted in economics,
offers a promising solution for intelligently allocating resources,
such as power and spectrum, in CR networks. There are different
approaches for implementing auction theory in wireless networks,
which have been investigated in~\cite{Niyato}. In general, in such
scenarios, users are rational and have their own strategies in order
to get more resources. Extensive existing works are available on
different auction approaches for spectrum allocation (e.g.,
see~\cite{Maharjan}). For instance, the authors in~\cite{JSAC2013}
find the maximization of the PUs' expected profit by proposing the
leasing based spectrum allocation for SUs. In addition, the first
price auction to optimize both the total payoff of SUs and revenue
of auctioneer is studied in~\cite{UncertainInfo}. One drawback of
the suggested scheme is that SUs might reveal wrong to further
improve their utilities. The work in~\cite{DySpan} provides a
spectrum allocation based upon a double-sided auction mechanism. In
this scheme, an untruthful behavior also brings suboptimal
solutions.

Competition among the selfish SUs is crucial to use rare resources
in the spectrum market framework~\cite{han}. More importantly,
non-cooperative users have intentions to cheat so as to gain more
benefits. The Vickrey Clarke Groves (VCG) auction mechanism is
commonly used in the auction games in order to provide not only the
assurance of truthfulness but also the maximization of the social
welfare~\cite{Vickrey}. For example, the authors
in~\cite{incentive1} and \cite{incentive2} proposed the incentive
mechanism to encourage users to contribute truthfully their
resources by forming coalitions. Moreover, because of selfishness of
SUs, each user attending in the auction has incomplete information
about the other users. Hence, selecting a proper learning task is a
big challenge for designing the distributed game. A Bayesian
nonparametric belief update scheme is suggested to solve this issue
in~\cite{Truthful}.

In CR networks, SUs are susceptible to several malicious attacks.
Several anti-attack mechanisms have been proposed in existing
literature~\cite{IUT}. For example, the problem of PU emulation
attack on CR networks has been investigated in~\cite{Reed} in which
a malicious user can send signals with the same PU transmission
characteristics in order to mislead the SUs. Instead, SUs can
recognize PUs' transmission by adapting a favorable verification
protocol. In addition, a game-theoretic approach based upon the
concept of secrecy capacity is proposed to model eavesdropping
attacks on CR networks in~\cite{Han2}. In~\cite{anti-jamming
stochastic}, a set of SUs is available in a stochastic medium and
they select randomized channel hopping as the defensive strategy.
This framework falls into the category of the zero sum stochastic
game and the authors propose a minimax-Q learning to find the
related solution. Besides, the randomized defense strategy for
channel hopping and power allocation with learning algorithms is
suggested in~\cite{Anti-Jamming Games}. However, in a spectrum
auction, users act selfishly and these defense strategies are not
fully applicable.

The main contribution of this paper is to jointly consider truthful
spectrum auction and the presence of a jamming attack. In this
scenario, two types of users exist: selfish SUs participating the
auction and a malicious jamming user that wishes to reduce the
social welfare as much as possible. Our key contributions can
therefore be summarized as follows:
\begin{itemize}
  \item To model the mentioned scenario, we formulated two inter-related games: a zero-sum stochastic
game between the CR network and the jammer, and an associated
mechanism design among the SUs at each stage of the game. Indeed,
the zero sum game exists between the CR network and the malicious
user, while mechanism design is considered among the SUs. Using our
proposed framework, the SUs do not show their selfishness and at the
same time cooperate with each other to get higher profits against
the malicious user.
  \item In order to realize the joint games, we propose an algorithm based on zero-sum game which can extensively reduce the complexity of
  solving the game with an asymmetric number of actions for the players.
  The proposition is a basis for the work because the malicious user and the SUs
  are unequal in the number of actions.
  \item Using 
    the derived proposition, we show that the zero-sum stochastic game and
    spectrum auction game can be converted to a centralized two-level
    spectrum auction in which SUs send their bids to a coordinator and the
    coordinator confronts against the malicious user. More specifically,
    the coordinator initially allocates spectrum to the first level
    bids, and then the remaining spectrum is allocated by the second auction.
    Indeed, the main idea of the centralized two-level auction is inspired from the randomized auction which is common in combinatorial
    auction theory such as~\cite{randomized auction 1} and \cite{randomized auction 2}. However,
    our considered scenario significantly differs from those existing works.
  \item A decentralized method based upon the centralized two-level
    auction is examined. The proposed algorithm use the proven interesting properties
    of the centralized game which extremely reduces the complexity of the game.
    Simulation results show that the loss in performance for the
    decentralized method in comparison with the centralized one is
    negligible.

  \item Due to the fact that SUs have no knowledge about the states of other SUs and
  jammer, the parameters for the decentralized scheme must be learnt
  from a proper scheme like the one proposed in~\cite{Basar}. We propose a Boltzmann-Gibbs algorithm to estimate the
  unknown parameters for each users. Simulation results show that this method
  yields considerable performance gains. Moreover, the convergence of the
  proposed decentralized game can be controlled by learning parameters.
\end{itemize}

The rest of this paper is organized as follows. The system model is
presented in Section II. In Section III, a centralized algorithm
based on a two-level auction is described. In Section IV, we propose
a truthful decentralized method in accordance with the proposed
centralized auction. The simulation results are given in Section V.
Finally, in Section VI, we conclude the paper.

\section{System Model and Protocol Description}
We consider a CR network consisting of ${M}$ channels having a
slotted-time structure indexed by ${j\in\{1,2,\ldots,M\}}$. Moreover, 
the duration of each time-slot is assumed to be ${T_s}$. There are
$N \ge M$ SUs that seek to access 
the vacant channels to send their data. Moreover, these users are
selfish and non-cooperative. The primary network consists of a
number of PUs who have a have priority to use the channels in a
slotted-time manner. We consider an on-off scheme to model the
channel usage, in which ${y_j} (t)=1$ and ${y_j}(t)=0$ indicate that
channel ${j}$ is idle and busy at time $t$,
respectively~\cite{anti-jamming stochastic} and \cite{Anti-Jamming
Games}. The transition probabilities from on-to-off and off-to-on
are ${\alpha}_{N2F, j}$ and ${\alpha}_{F2N, j}$, respectively.
Without loss of generality, we assume that every SU can only use one
channel at time $t$~\cite{van der Schaar}. In order
to avoid the conflict with the 
PUs transmission, each SU knows the
availability of all the channels before transmitting. 
This can be done by using wideband sensing or cooperative sensing
techniques~\cite{CR_architecture}.

\begin{figure}[t]
  \begin{center}
    \includegraphics[height=4 in, width=5 in]{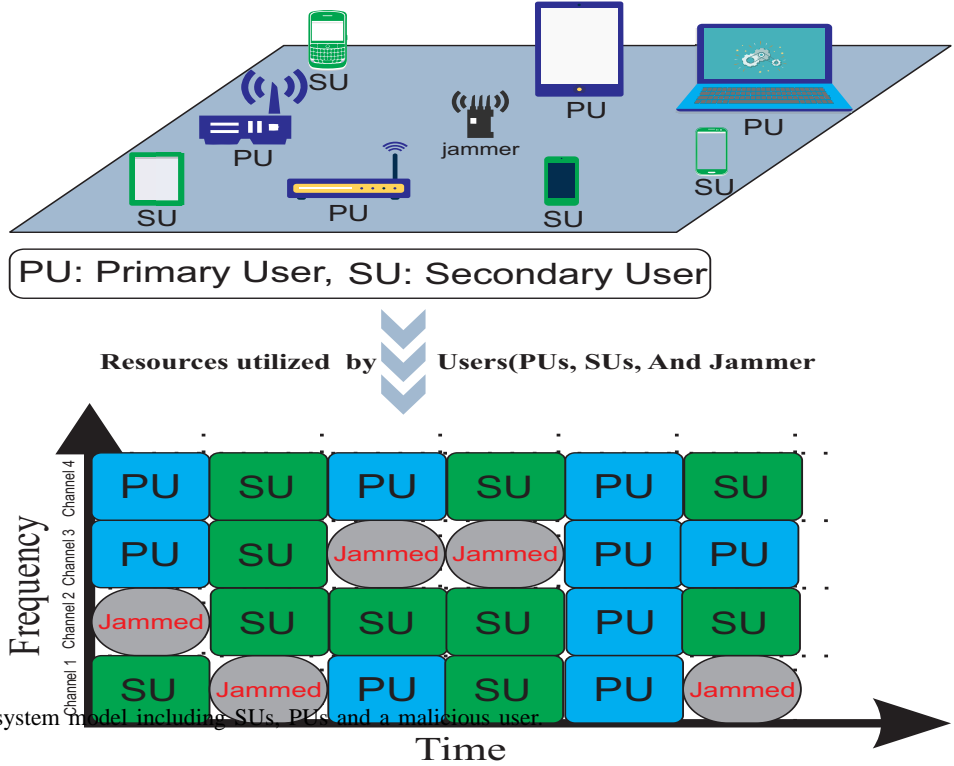} 
    \begin{center}\vspace{-2 cm}
    \caption{The system model including SUs, PUs and a malicious user.} \vspace{-2 cm}
    \end{center}
    \label{SystemModel}
  \end{center}
\end{figure}
The state of channel $j$ 
for SU $i$ is assumed to be the
received signal-to-noise-ratio (SNR) $\gamma_{ij}(t)$, following
an exponential distribution with mean of
$\overline{\gamma_{ij}}$. Similar to~\cite{Finite-state}, we
represent $\gamma_{ij}(t)$ by discrete states to attain a finite
Markov chain. In addition, let $b_i^t$ indicate the buffer state of
user $i$ at time $t$ and $b_i^t\in\{0,1,\ldots,B_{\max}\}$ where
$B_{\max}$ is the maximum buffer size. Thus, the state of SU $i$ at
time $t$ is
$\boldsymbol{s}_i(t)=\big(\gamma_{i1}(t),\gamma_{i2}(t),\ldots,\gamma_{iM}(t),b_i^t\big)$
and the state of the stochastic game is described as follows:
\begin{equation}
\boldsymbol{S}(t)=\big(y_1(t),\ldots,y_M (t),\boldsymbol{s}_1
(t),\ldots,\boldsymbol{s}_N (t)\big),
\end{equation}
where the state of the game $\boldsymbol{S}(t)$ consists of the
state of each SU and the occupancy state of each channel. The
assigned channel to the $i$-th SU is denoted by $A_i (t)$. Moreover,
it is possible that no channel is assigned to the 
SU, i.e., $A_i(t)=0$. Thus, we have $A_i(t)\in\{0,1,\ldots,M\}$.

Assume there is a malicious attacker in this scenario which attempts
to interrupt the communication links of 
the SUs by inserting interference. The action of malicious user is
to jam $L$ channels chosen from the vacant channels. Indeed, if the
malicious user jams channel $j$, the communication link is assumed
to be disrupted at that time. We assume that the jammer knows the
channel occupancy states at each stage time. For simplicity, we
assume $L=1$, and our approach can be extended to $L>1$ case. The
action of jammer, $A_0(t)\in\{1,2,\ldots,M\}$, indicates the jammed
channel by the attacker. Fig.~1 shows the proposed system model and
illustrates how users occupy the time-frequency resources.

Notice that the availabilities of the channels are only imposed by
PUs, and hence, they are independent of the attacker's action and
SUs' actions. Consequently, we can now derive the transition
probability of the states as
\begin{eqnarray}
P\big(\boldsymbol{S}(t+1)\mid
\boldsymbol{S}(t),A_0(t),A_1(t),\ldots,A_N
(t)\big)=\qquad\qquad\qquad\qquad\qquad\\
P\big(y_1(t+1),\ldots,y_M(t+1)\mid y_1(t),\ldots,y_M(t)\big)
\prod_{i=1}^N{P\big(\boldsymbol{s}_i(t+1)\mid
\boldsymbol{s}_i(t),A_0(t),\ldots,A_N(t)\big)},\nonumber
\end{eqnarray}
$\boldsymbol{s}_i(t+1)$ includes information about the channels'
conditions and the buffer state. The channel conditions do not
depend on the SUs action. Besides, the buffer state, $b_i(t+1)$, is
affected by the jammer action, $A_0(t)$, the action of SU $i$,
$A_i(t)$, and $\boldsymbol{s}_i(t)$. Hence, we can express the last
term of (2) as
\begin{eqnarray}
P\big(\boldsymbol{s}_i(t+1)\mid \boldsymbol{s}_i(t),A_0(t),A_1(t),\ldots,A_N(t)\big)=P\big(\boldsymbol{s}_i(t+1)\mid \boldsymbol{s}_i(t),A_0(t),A_i(t)\big)\nonumber\\
= P\big(b_i(t+1)\mid b_i(t),A_0(t),A_i(t)\big)
\times\prod_{i=1}^M{P\big(\gamma_{ij} (t+1)\mid
\gamma_{ij}(t)\big).}\qquad\qquad\qquad\,\,\,\,
\end{eqnarray}

We denote the incoming traffic of SU $i$ at time $t$ as $f_i^t$
where $f_i^t\in\{0,1,\ldots,\infty\}$. It is assumed that $f_i^t$
has the Poisson distribution with the average
$\overline{f_i}$~\cite{van der Schaar}. Moreover, the buffer state
is derived from $b_i(t+1)=\min\big((b_i(t)- g_{A_i, A_0}(t) )^+
+f_i^t,B_{\max}\big)$. Hence, we have the following expression for its
transition probability
\begin{eqnarray}
P\big(b_i(t+1)\mid
b_i(t),A_0(t),A_i{(t)}\big)=\qquad\qquad\qquad\qquad
\end{eqnarray}
\vspace{-.75 cm}
\begin{eqnarray*}
 \left\{
\begin{array}{ll}\frac{\overline{(f_i)^x} e^{\overline{f_i}}}{x!}, &\textrm{$0\leq x<-(b_i(t)-g_{A_i, A_0}(t))^+ +B_{\max}$},\\
\sum_{x=B}^\infty {\frac{\overline{(f_i )^x} e^{\overline{f_i}}}{x!}}, &\textrm{$x=-\big(b_i(t)-g_{A_i, A_0}(t)\big)^+ +B_{\max}$},\\
\end{array}\right.\nonumber
\end{eqnarray*}
where $(c)^+=\max(c,0)$ and $g_{A_i, A_0}(t)$  indicates the
transmission bit rate if channel $A_i(t)$ is selected and channel
$A_0(t)$ is jammed. Therefore, $g_{A_i, A_0}(t)$ can be calculated
as \cite{gol97_2}
\begin{eqnarray}
g_{A_i, A_0}(t)= \bigg\lfloor{T_sW \log_2\big(1+\frac {1.5
\gamma_{i,j}}{\ln(\frac {0.2}{\text{BER}_{\text{tar}}})}
\big)}\bigg\rfloor I(A_i\neq A_0),
\end{eqnarray}
where $T_s$, $W$ and $\text{BER}_{\text{tar}}$ are the time
duration, bandwidth of each channel and target bit error rate,
respectively. In (5), $\lfloor{X}\rfloor$ and $I(Y)$ indicate the
largest integer number which is lower than $X$ and the sign of $Y$,
respectively. When the $i$-th SU selects channel $A_i(t)$ and the
jammer selects the $A_0(t)$-th channel at the same time, the utility
function of user $i$ at time $t$ is characterized as follows
\begin{eqnarray}
r_i\big(\boldsymbol{S}(t),A_i(t),A_0(t)\big)= -\Big(b_i(t)-g_{A_i,
A_0}(t)-B_{\max}+f_i^t\Big)^+.
\end{eqnarray}

In our scenario, we consider the presence of a coordinator that
allocates spectrum to the SUs according to the submitted bids while
maximizing the worst-case social welfare corrupted by the attacker.
Hence, the interactions between the coordinator and the SUs are cast
as an \emph{auction} with the following elements:

\begin{itemize}
\item The auctionees are the SUs which aim at using the vacant channels.
\item The auctioneer is the coordinator which allocates the channels to SUs.
Afterwards, the auctioneer and coordinator are used interchangeably.
\item Each bid is denoted by $a_{ij,k}$, where $1\leq j,k\leq
M$. Here, $a_{ij,k}$ indicates the proper bid for SU $i$ to use
channel $j$ while the attacker jams channel $k$.
\item The following
constraints must be satisfied at each stage of the auction:
\begin{eqnarray}
\sum_{j=1}^M z_{ij}(t)\leq 1\qquad\qquad\nonumber\\
\begin{array}{ll} {\sum_{i=1}^N z_{ij}(t)= 1}, &\textrm{if channel $j$ is idle},\\
{\sum_{i=1}^N z_{ij}(t)= 0}, &\textrm{if channel $j$ is busy},\\
\end{array}
\end{eqnarray}
in which $z_{ij}(t)\in\{0,1\}$ shows that channel $j$ is allocated
to the $i$-th SU if $z_{ij}(t)=1$; and is not allocated otherwise.
\end{itemize}
In order to combat the 
jammer, the coordinator should 
assign the channels to the SUs via a random strategy.
In the next section, we will 
investigate 
 this optimal strategy.

\section{Anti-Jamming Decentralized Game Based on Learning Process}
In the previous section, the PC-game is proposed in order to extract
the anti-jamming mechanism under the condition that all SUs and the 
  auctioneer act as one player to defeat the malicious user. However,
this assumption may not hold in general since the SUs are selfish
and maybe untruthful. Unreliable information may lead to an improper
strategy for protection of the 
SUs against the jammer. Besides, the SUs
send their $M'^2$ bids to the coordinator, which has the high
complexity. Due to these drawbacks, this section suggests a
decentralized method according to the framework provided by the PC-game.

In the PC-game, we use a two level auction, and our aim is to
specify a distribution function to the actions. These actions can be
recognized by the first and second preferences of all the SUs. First,
pay attention to 
$\boldsymbol{p}_1^{*T}\boldsymbol{U}\boldsymbol{p}_2^{*}=\bigg(\sum_{l=1}^{\frac{N!}{(N-M')!}}\boldsymbol{p}^{*}_{1,l}\boldsymbol{U}^{l}\bigg)\boldsymbol{p}^{*}_2$
where $\boldsymbol{p}_1^*$ and $\boldsymbol{p}_2^*$ are the optimal policies of the auctioneer and the jammer, respectively. 
Moreover, $\boldsymbol{p}^{*}_{1,l}$ and $\boldsymbol{U}^{l}$ are  
 the $l$-th entry of $\boldsymbol{p}_1^*$ and the $l$-th row of payoff
matrix $\boldsymbol{U}$ of the original game in \textit{Definition
1}. If we extend each $\boldsymbol{U}^{l}$ into its elements, we
have the following formulation: \vspace{0 cm}
\begin{eqnarray}
\bigg(\sum_{l=1}^{\frac{N!}{(N-M')!}}\boldsymbol{p}^{*}_{1,l}U^{l}\bigg)\boldsymbol{p}^{*}_2=\bigg(\sum_{i=1}^N\sum_{j=1}^{M'}
\boldsymbol{p}_{u(i,j)}^*[a_{i,j,1},\ldots,a_{i,j,M'}]\bigg)\boldsymbol{p}^{*}_2,\vspace{-.5
cm}
\end{eqnarray}
in which $\boldsymbol{p}_{u(i,j)}^*$ is equal to the probability of
selection of the $j$-th channel for the $i$-th user. Every policy, 
which yields 
 the same $\boldsymbol{p}_{u(i,j)}^*$, is the optimal
strategy against the attacker. This fact motivates us to move from
the PC-game to a 
 distributed game. In the PC-game, we specify a
probability to each action distinguished by the first auction or
equivalently the first preferences of the SUs. By truthfulness
assumption and help of the mentioned fact, if each SU individually estimate the probabilities connected with the preferences over the channels, then the value of the PC-game obtained from (14) can be approximated by the following formulation:

\vspace{0 cm}
\begin{eqnarray}
\sum_{l_1=1}^{M'}\ldots\sum_{l_N=1}^{M'} Q_{l_1}\ldots Q_{l_N}
\boldsymbol{U}_{l_1,\dots,l_N}\boldsymbol{p}_2^{*}\approxeq
\boldsymbol{p}_1^{*T}\boldsymbol{U}\boldsymbol{p}_2^{*},\vspace{-.5
cm}
\end{eqnarray}
where $Q_{l_i}$ and $\boldsymbol{U}_{l_1,\dots,l_N}$ are the
estimated probability related to the first preference by the $i$-th
SU
 and the value of the game when the SUs's preferences are
${l_1,\dots,l_N}$, respectively.

Each auction consists of $M'$ allocations to the SUs. Note that from
\textit{Proposition $1$}, we only need $M'$ auctions to reach to the
best response against the jammer. Thus, there are at most ${M'}^2$
important probabilities, $\boldsymbol{p}_{u(i,j)}^*$, at each stage
of game.
 Moreover, it can be easily demonstrated that every policy, which has these $M'^2$ probabilities, is optimal from the perspective of the zero-sum game. On the other side, each SU has control over $M'$ probabilities for stating its first preference over the channels. From this point of view, the SUs have $N \times M'$ variables for estimations of $M'^2$ important probabilities which are improved with increasing $N$ compared to $M$.

At this time, by applying the auction feature to the game, the coordinator can
get payments from the 
SUs. The payment of each SU is constructed from two parts. One
payment part is related to the
first-auction and the other part 
is associated with the second-auction. The computation approach of the payment for the
first-auction  which is similar to~\cite{van der Schaar}is stated as
\begin{eqnarray}
p_i^t=\sum_{(k=1,k\neq i)}^N \sum_{j=1}^{M'} z_{kj}^{t,opt} a_{kj}'
(t) -\max_{(z_{kj} \mid  {a'_{ij}=0},  \forall  j)} \sum_{(k=1,k\neq i)}^N
\sum_{j=1}^{M'} z_{kj}^{t} a_{kj}'(t),
\end{eqnarray}
in which $z_{kj}^{t,opt}$ is the solution of the first auction.
 For the second-auction, this payment can also be computed by the same procedure while the selected SUs 
 in the first-auction and their corresponding announced bids are omitted by the coordinator.  The PD-game
procedure is described in Table~\ref{tab:table 2}. We show these
payments oblige the SUs to bid truthfully. In order to prove that
the proposed distributed game (PD-game) contains the truthful
mechanism, first we define the concept of truthfulness in
expectation.
\begin{definition}
Assume $v_i$, $v_i'$, $v_{-i}$ and $p_i^t$ are the real value of bid for user $i$, the announced value of bid for user $i$, the value of bids for other users and the payment assigned to user $i$,
respectively. A mechanism is truthful in expectation when for any user $i$ and any $v_{-i} \in V_{-i}$ of other users, the expectation of profit attained by user $i$, $\mathbb{E}\{{v_i}'-{p_i^t}({v_i}',v_{-i})\}$ is maximum if $v_i'=v_i$~\cite{Nisan}.
\end{definition}

We now focus on a proposition which states that the PD-game is
truthfulness in expectation.

\begin{proposition}
The proposed procedure for assigning payment satisfies truthfulness
in the expectation criterion.
\end{proposition}
\begin{proof}
The proof is given in Appendix C.
\end{proof}


Note that the payment of each SU, which is dependent on all the SUs' bids, converts the profit gained by each SU into a notion of the overall value of the zero-sum game. Thus, we are trying to model the game between each SU and the attacker as the zero-sum game separately so that the separate game for each SU has some external factors related to 
 other SUs, and each SU is effective only on
a certain amount of the profit.

By doing so, every SU computes the distribution of stating its preference over the channels. In addition, the communication burden of 
 stating its bids obviously plummets. 
  Since, the SU only
sends $M'$ bids instead of stating $M'^2$. Duties of the coordinator
decreases since it 
 only computes the first and second-auctions and their related payments. Indeed,
the utility matrix of the separate game between each SU and the attacker is modeled as a
$(M'\times M')$ 
 matrix because the SU has $M'$ choices for the announcement of its first preference.
\begin{table}
\centering
\renewcommand{\arraystretch}{2}
\caption{The proposed decentralized game} \label{tab:table 2}
\centering
\begin{tabular}{l}
\hline \textbf{Step 1}. The SUs submit the bid based upon (12) to the
coordinator. At the same time, the SUs announce their preferences over\\
channels in order to be used in the first and second auctions.\\
\textbf{Step 2}. First auction is computed for the first preferences of
the SUs. Then, allocation and payment for each SU is assigned to\\ them
by using (7) and (16).\\
\textbf{Step 3}. Similarly, the second
auction is computed for the remaining channels and the SUs.\\
\hline\vspace{-1 cm}
\end{tabular}
\end{table}
Note that our algorithm is distinct from work suggested in~\cite{Multiagent planning} in which authors employ a factored approximation of the overall Q-function based upon the linear combination of users' Q-function for the stochastic game. The proposed algorithm is not applicable in our scenario because the SUs are selfish and interested in benefiting further. Indeed, the payment structure makes the profit of SUs' network directly relevant to each individual profit due to \textit{Proposition $3$}. Instead, $p^{*}_1$ is estimated by SUs' probabilities, $Q_{l_1}, \ldots, Q_{l_N}$.

The fundamental difficulty of the PD-game is that each SU does not
know enough about its related separate utility matrix. Remembering 
 that the 
  game will be repeated infinitely, and therefore, the SUs can learn their utilities by a certain learning scheme. We employ the scheme 
   proposed in~\cite{Basar}. The advantage of this scheme is that each SU can adapt different patterns of learning.   The 
 probabilistic strategy over the actions and utility of each stage can be learned through the game. 
  First, we apply an iterative Boltzmann-Gibbs strategy which is stated as
\begin{eqnarray}
\boldsymbol{\sigma}_i \big(\boldsymbol{q}_{1i}^t,
\widehat{\textbf{u}}_{1i,t}, \boldsymbol{S}(t)\big)(j)=
\frac{\boldsymbol{q}_{1i}^t\big(j,\boldsymbol{S}(t)\big)
e^{\frac{\widehat{\textbf{u}}_{1i,t}\big(j,\boldsymbol{S}(t)\big)}{\epsilon}}}{\sum_{j=1}^{M'}
\boldsymbol{q}_{1i}^t\big(j,\boldsymbol{S}(t)\big)
e^{\frac{\widehat{\textbf{u}}_{1i,t}\big(j,\boldsymbol{S}(t)\big)}{\epsilon}}}\vspace{-1
cm}
\end{eqnarray}
where $\boldsymbol{q}_{1i}^{t}\big(j, \boldsymbol{S}(t)\big)$ and $\widehat{\textbf{u}}_{1i,t}$ are distribution of selecting channel $j$ as the first preference of SU $i$ and the estimated average payoffs updated at 
 iteration $t$, respectively~\cite{Basar}. Next, we update distribution and payoff, respectively, as
\begin{eqnarray}
\boldsymbol{q}_{1i}^{(t+1)}\big(\boldsymbol{S}(t)\big)=(1-\lambda_{1i,t})\boldsymbol{q}_{1i}^t\big(\boldsymbol{S}(t)\big)
+\lambda_{1i,t}\boldsymbol{\sigma}_i\big(\boldsymbol{q}_{1i}^t, \widehat{\textbf{u}}_{1i,t}, \boldsymbol{S}(t)\big)(j)\qquad\qquad\\
\widehat{\textbf{u}}_{1i,t+1}\big(\boldsymbol{S}(t)\big)={\widehat{\textbf{u}}}_{1i,t}\big(\boldsymbol{S}(t)\big)+\frac{\mu_{1i,t}}{\boldsymbol{q}_{1i}^t
\big(j,\boldsymbol{S}(t)\big)}
\bigg(U_{1i,t}\big(\boldsymbol{S}(t)\big)-\widehat{\textbf{u}}_{1i,t}\big(\boldsymbol{S}(t)\big)\bigg).
\end{eqnarray}
in which $U_{1i,k,t}$ is the profit gained by SU $i$ at time $t$
when selecting channel $j$ as its preference, which is zero when no
channel is assigned to it, and is $a_{{ij},k}'-p_{i}(t)$ when
channel $j$ is assigned.  Furthermore, $\mu_{1i,t}$ and $\lambda_{1i,t}$
are the learning rates indicating players' capabilities of
information retrieval and update. Therefore, each SU can learn the distribution over its preference from  implementing a Q-learning based method.  It can be proved that Q-learning method converges to the optimal solution for only single-agent case; However, there is no such a guarantee for multi-agent cases~\cite{Iterative stochastic process}. In the next section, simulation results illustrate the convergence of the PD-game to the sub-optimal solutions.

\vspace{.5 cm}\section{Simulation results}\vspace{.5 cm} In this
section, we provide simulation results to verify the truthful
anti-jamming network. We consider a cognitive radio environment with
$M$ channels, $N$ secondary users and a malicious user. We assume
that the state of signal to noise ratio for SU $i$ and channel $j$,
$\gamma_{ij}$, has three values ${10, 30}$ and $50$. The probability
of state transitions from these states are
$p(\gamma_{ij}=10|\gamma_{ij}^1=10)=0.4$,
$p(\gamma_{ij}=30|\gamma_{ij}^1=10)=0.3$,
$p(\gamma_{ij}=10|\gamma_{ij}^1=30)=0.3$,
$p(\gamma_{ij}=30|\gamma_{ij}^1=30)=0.4$,
$p(\gamma_{ij}=10|\gamma_{ij}^1=50)=0.3$, and
$p(\gamma_{ij}=30|\gamma_{ij}^1=50)=0.3$. In addition,
${\alpha}_{N2F, j}=0.3$ and ${\alpha}_{F2N, j}=0.4$ for $1 \leq
j\leq M$. We set also $\text{BER}_{\text{tar}}$ for all the users in
(5) as $10^{-5}$.
\subsection{Convergence}
 The convergence speed of the PC-game and the PD-game for three SUs are investigated in Fig.~2 and
Fig.~3 when $M=2$ and $M=3$, respectively. Besides, $B_{\max}=2$ and
$\overline{f_i}= 0.5$ for all SUs for the either case. The
normalized cumulative value of SUs is used as a convergence
comparison tool. As Fig.~2 and Fig.~3 report, both algorithms
converge; however, the PD-game takes longer time to reach the stable
solution. The PD-game is done in the decentralized scheme with
incomplete information. Therefore, it needs more times to learn the
unknown parameters. In particular, the convergence rates in Fig.~3
for both the PC-game and PD-game are quite slower than those in
Fig.~2. Indeed, increase in $M$ leads to rises in the numbers of the
states and the complexity of the system. Consequently, the required
numbers of iterations in Fig.~3 explicitly becomes greater.
\begin{figure}[t]
  \begin{center}
    \includegraphics[height=2.5 in, width=3.6 in]{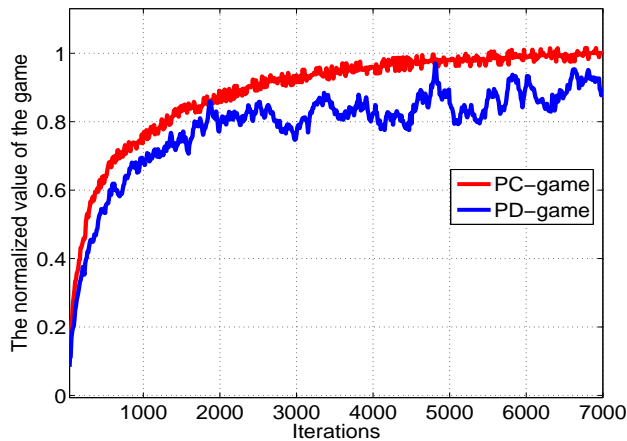} 
    \begin{center}\vspace{-.2 cm}
    \caption{The convergence of the normalized cumulative value of SUs in the PD-game and PC-game in a networks with $M=2$ and $N=3$.} \vspace{-.2 cm}
    \end{center}
    \label{fig 2}
  \end{center}
\end{figure}

\begin{figure}[t]
  \begin{center}
    \includegraphics[height=2.5 in, width=3.6 in]{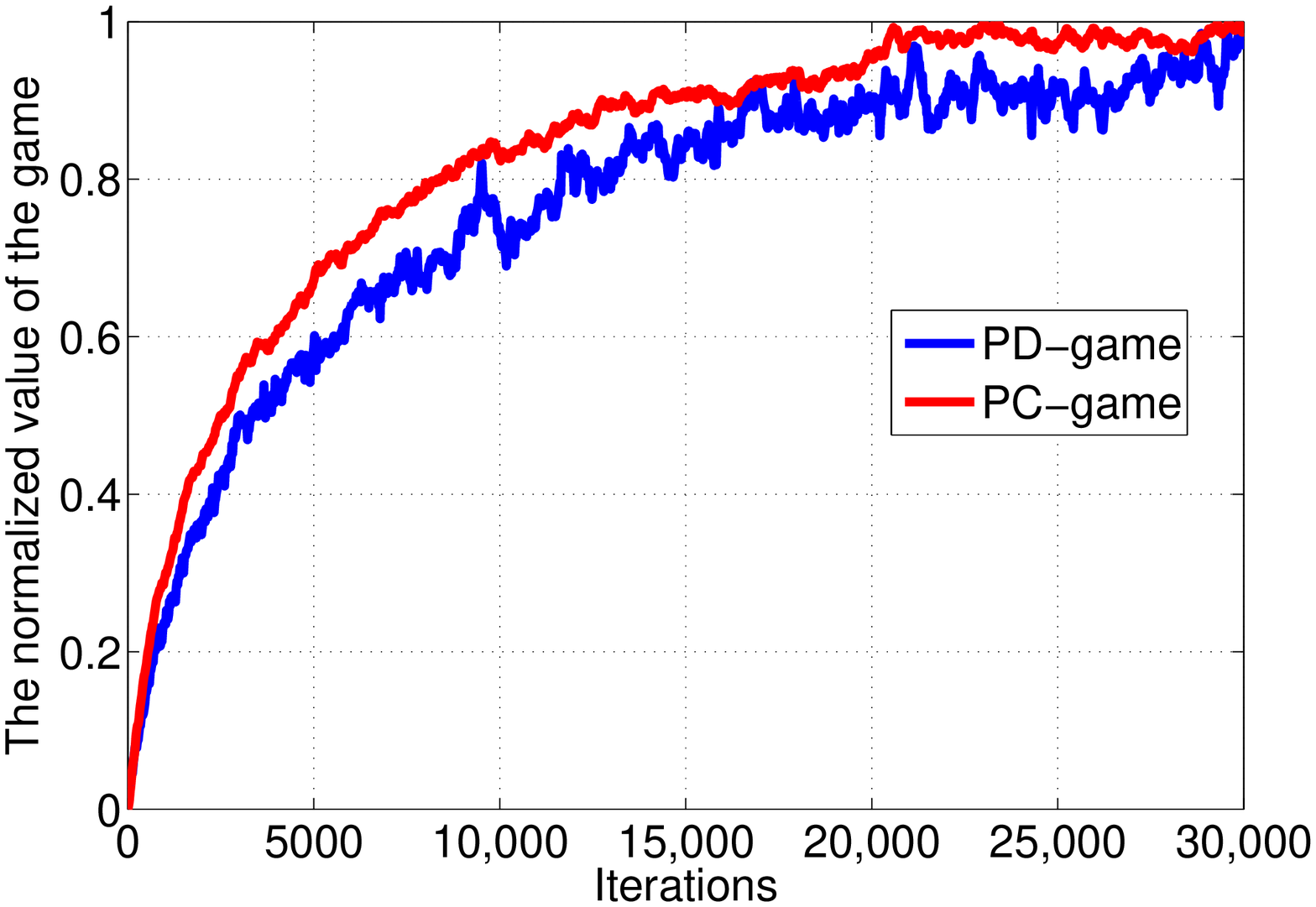} 
    \begin{center}\vspace{-.2 cm}
    \caption{The convergence of the normalized cumulative value of SUs in the PD-game and PC-game in a networks with $M=3$ and $N=3$.} \vspace{-.2 cm}
    \end{center}
    \label{fig 2}
  \end{center}
\end{figure}

\begin{figure}[t]
  \begin{center}
    \includegraphics[height=2.5 in, width=3.6 in]{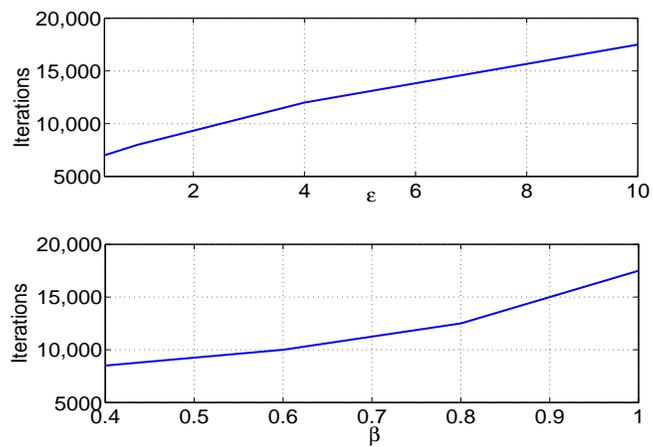} 
    \begin{center}\vspace{-.2 cm}
    \caption{The effect of different $\beta$ and $\epsilon$ on the performance of the PD-game.} \vspace{-.2 cm}
    \end{center}
    \label{fig 4}
  \end{center}
\end{figure}

The learning parameters $\lambda_{1i,t}$, $\mu_{1i,t}$ and
$\epsilon$ in (17), (18) and (19) play important roles in the
convergence of the PD-game. In~\cite{Basar}, it is shown that
$\frac{\lambda_{1i,t}}{\mu_{1i,t}}\rightarrow 0$ for  assurance of
the convergence. Hence, we consider
$\frac{\lambda_{1i,t}}{\mu_{1i,t}}=\frac{\frac{1}{T_{\boldsymbol{S}}^{(1+\beta)}}}{\frac{1}{T_{\boldsymbol{S}}}}$
where $T_s$ is the repetition numbers of state occurrence, where
$\beta>0$. Fig.~4 depicts the effect of different $\beta$ and
$\epsilon$ on the iterations required for the convergence under the
mentioned condition when $M=2$. It is clear that when these
parameters increase, the convergence speed decrease, since the
impact of instantaneous utilities on current strategy decreases.

%
%

\subsection{The effects of SU parameters on performance}
In this part, the effects of the maximum allowable $B_{\max}$, the
number of channels $M$, and the number of users $N$ on the PD-game
and the PC-game are evaluated. In order to have a similar benchmark
for comparison of two methods, we define a new parameter $\theta$
based on (6) as,
\begin{eqnarray*}
\theta=\sum_{t=0}\sum^N_{i=1} -r_{i}(t)/N.
\end{eqnarray*}
Fig.~5 and Fig.~6 illustrate the performance of the PC-game and the
PD-game by $\theta$ for variable $B_{\max}$ and $N$ when $M=2$ and
$M=3$, respectively. The other parameters are set alike to the
previous part. In Fig.~5, the SU with the greater $B_{\max}$ is able
to hold the data for a longer time. Thus, the increment in
$B_{\max}$ decreases $\theta$. In other words, it can improve the
performance of the system. However, increase in $N$ has opposite
impact on the $\theta$ which is result of increasing the dropping
probability of data. Moreover, Fig.~6 shows the performance when
$M=3$. Note that both the PC-game and PD-game in Fig.~6 have lower
$\theta$ rather than those in Fig.~5 for the same condition. Indeed,
$M=3$ increases the opportunities of available vacant channels for
each SU; therefore, decreases the numbers of unsent buffered
information.
\begin{figure}[t]
  \begin{center}
    \includegraphics[height=2.5 in, width=3.6 in]{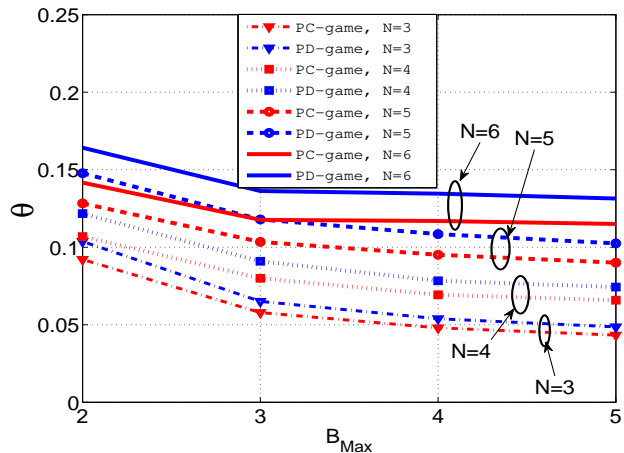} 
    \begin{center}\vspace{-.2 cm}
    \caption{The effect of different $B_{\max}$s and $\epsilon$s on the performance of the PD-game.} \vspace{-.2 cm}
    \end{center}
    \label{fig 4}
  \end{center}
\end{figure}

\begin{figure}[t]
  \begin{center}
    \includegraphics[height=2.5 in, width=3.6 in]{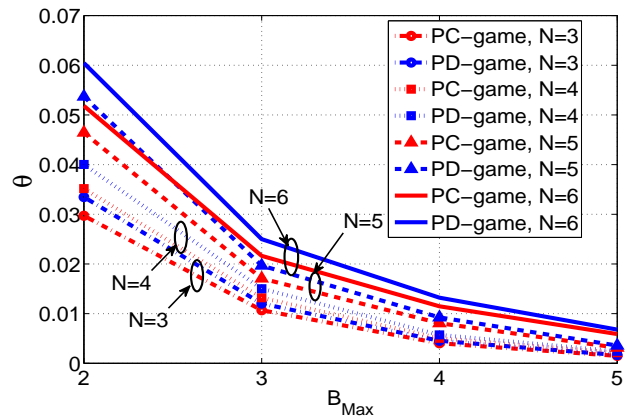} 
    \begin{center}\vspace{-.2 cm}
    \caption{The effect of different $B_{\max}$ and $N$ for $M=3$ on the performance of the PD-game.} \vspace{-.2 cm}
    \end{center}
    \label{fig 6}
  \end{center}
\end{figure}

\begin{figure}[t]
  \begin{center}
    \includegraphics[height=2.5 in, width=3.6 in]{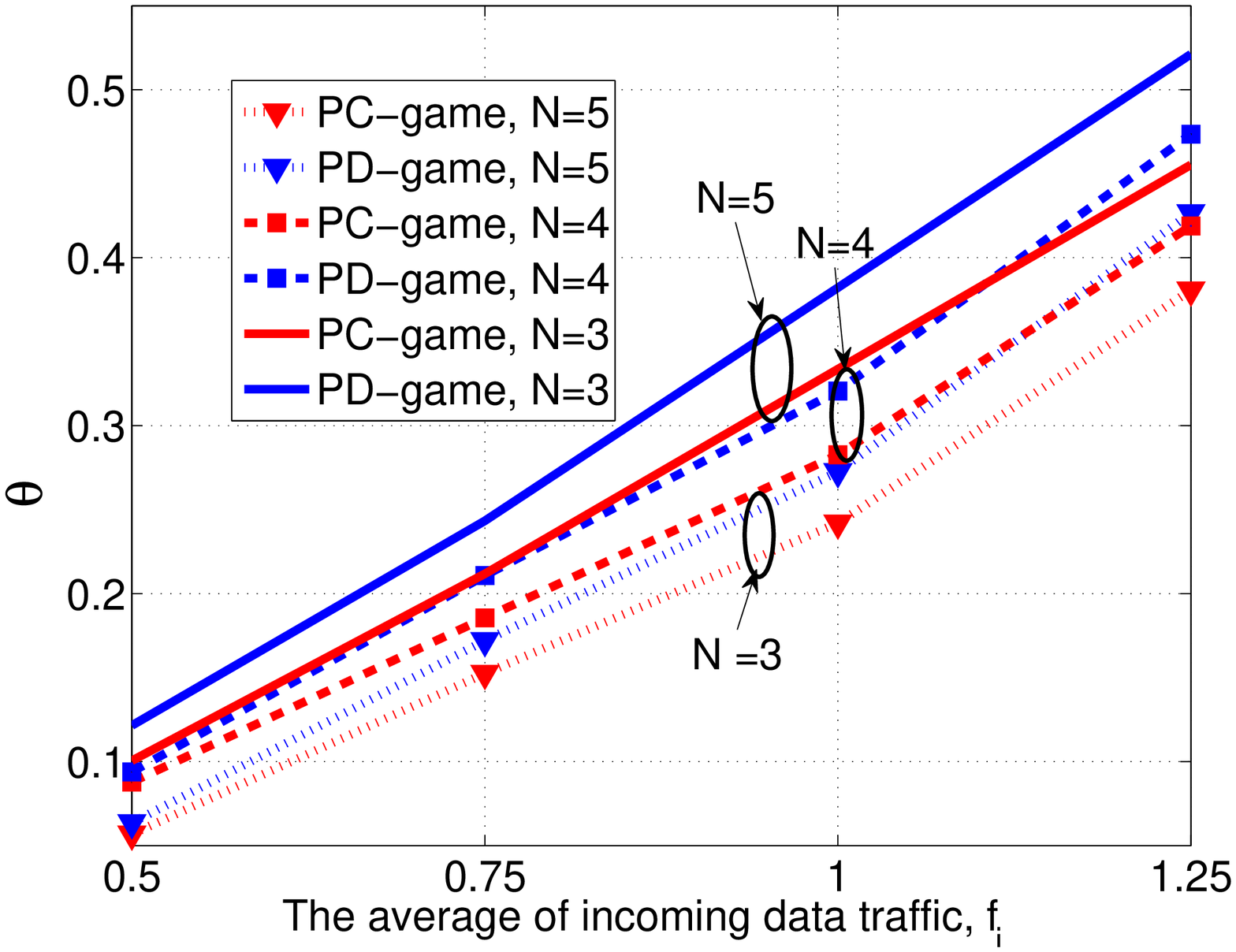} 
    \begin{center}\vspace{-.2 cm}
    \caption{The effect of different $\overline{f_i}$ and $N$ for $M=2$ on the performance of the PD-game and the PC-game.} \vspace{-.2 cm}
    \end{center}
    \label{fig 2}
  \end{center}
\end{figure}

\begin{figure}[t]
  \begin{center}
    \includegraphics[height=2.5 in, width=3.6 in]{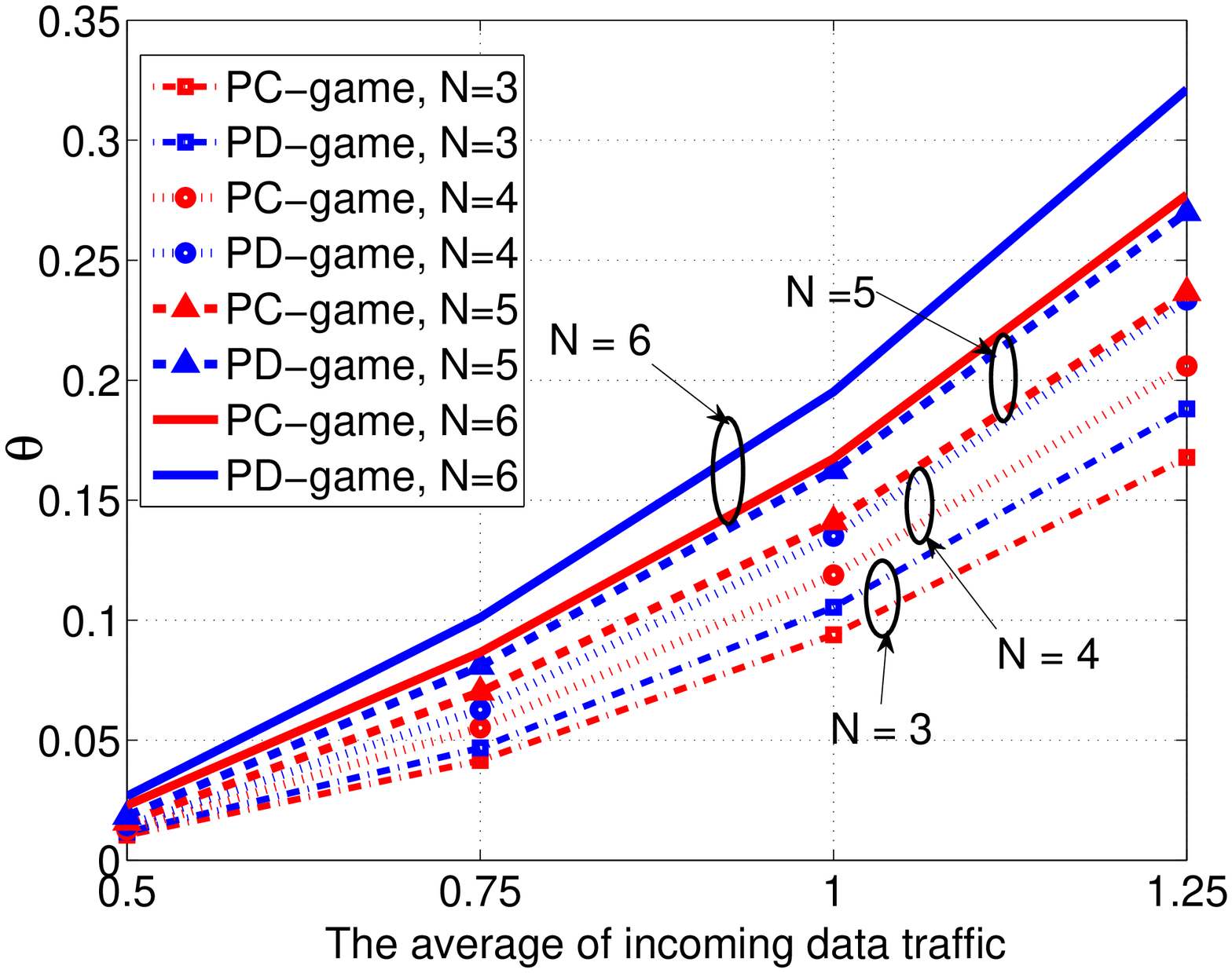} 
    \begin{center}\vspace{-.2 cm}
    \caption{The effect of different $\overline{f_i}$ and $N$ for $M=3$ on the performance of the PD-game and the PC-game.} \vspace{-2 cm}
    \end{center}
    \label{fig 2}
  \end{center}
\end{figure}

\begin{figure}[t]
  \begin{center}
    \includegraphics[height=2.5 in, width=3.6 in]{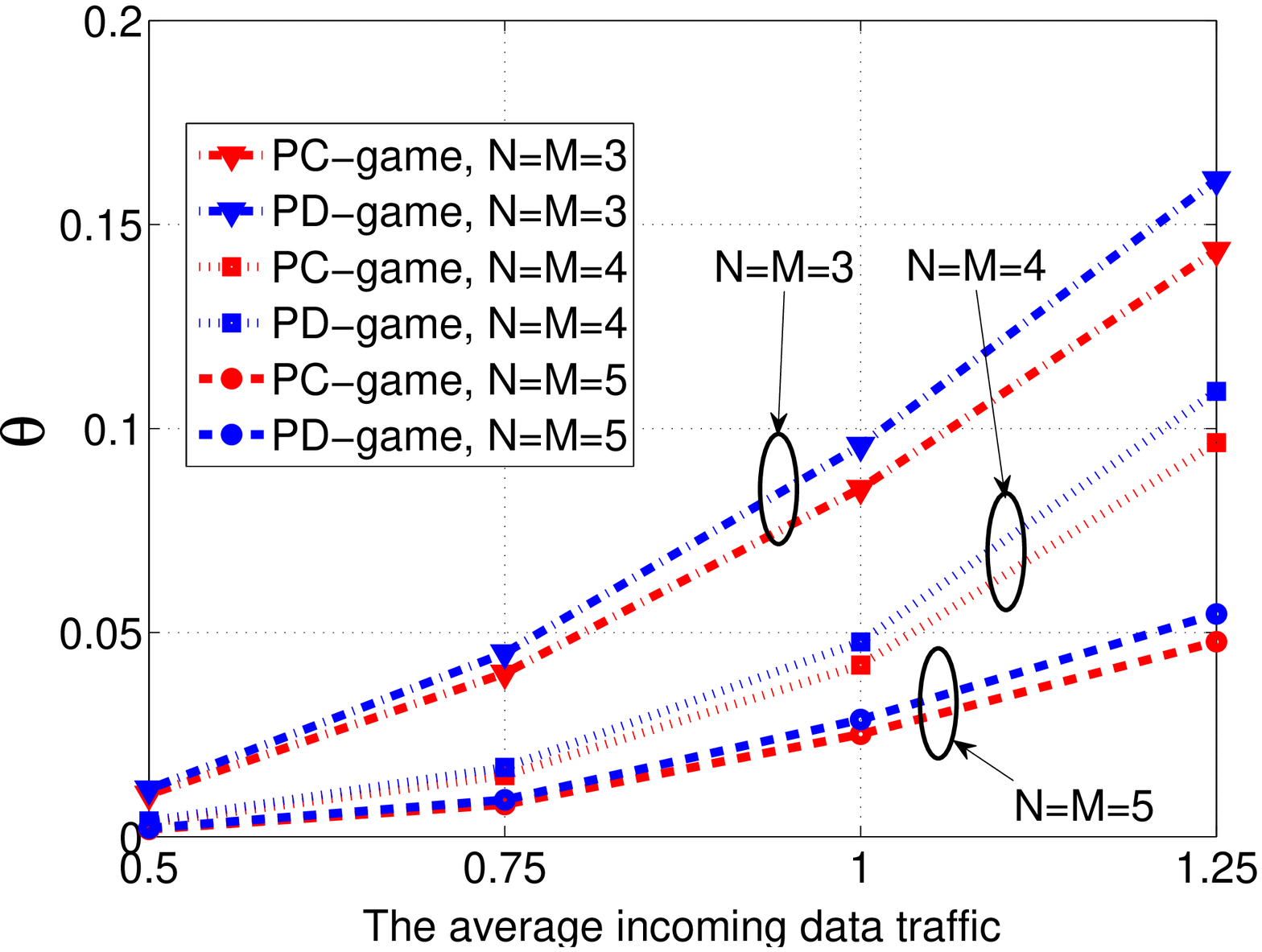} 
    \begin{center}\vspace{-.2 cm}
    \caption{The effect of different $\overline{f_i}$ on the performance of the PD-game.} \vspace{-2 cm}
    \end{center}
    \label{fig 2}
  \end{center}
\end{figure}

The performance of the scenario for different average of incoming
traffic $\overline{f_i}$ and the numbers of SUs is shown in Fig.~7
and Fig.~8. The results are obtained for $M=2$ and $M=3$,
respectively. Rise in $\overline{f_i}$ means that the average of
incoming traffic increase. The outcome of the rise is to receive
more traffic data at each stage of the game; as a result, the
average unsent traffic $\theta$ increase. Finally, Fig.~9 displays
$\theta$ versus $\overline{f_i}$ when $N=M$. Notice that increase in
$N$ along with $M$ causes $\theta$ to be lower which validates our
discussion about the performance of the scheme.

%





%
\vspace{-1 cm}
\section{Conclusion}
Spectrum management among the SUs is a vital issue for CR networks,
and auction theory provides a helpful tool to allocate spectrum to
SUs. In this article, first, we proposed a centralized two-level
auction which combined both the advantages of efficient resource
assignment to SUs and acting against the malicious user. Next, a
proposition for the zero-sum game was given which can be applied in
a game with the non-uniform number of users' actions. More
importantly, we introduced a decentralized protocol based upon the
centralized method properties and the mentioned proposition. The
decentralized scheme obliges SUs to bid truthfully because SUs can
gain higher profit in expectation for the long-term interaction.
Simulation studies show that both the centralized and decentralized
scheme converge in the limited numbers of stages. Moreover, the
performance of the proposed approach are comparable with the
efficient centralized solution. 
 \appendices
\section{Proof of Proposition 1}
Consider a zero-sum game with payoff matrix $\boldsymbol{O}$ as follows 

\[
 \boldsymbol{O} =
\begin{pmatrix}
o_{1,1}  & \cdots & o_{1,l_2} \\
\vdots & \ddots & \vdots \\
o_{l_1,1} & \cdots & o_{l_1,l_2}
\end{pmatrix}
\] 
\begin{equation}
\end{equation}
in which $o_{n,m}$ shows that player $1$ and player $2$ obtain 
${o_{n,m}}$ and ${-o_{n,m}}$ profit
 when they select their 
  $n$-th and $m$-th actions, respectively. To attain the optimal solution~\cite{non_game}, we should consider mixed strategy with the help of the following equation:
\begin{equation}
\max_{\boldsymbol{p}_1}\min_{\boldsymbol{p}_2} \boldsymbol{p}_1^T
 \boldsymbol{O}\boldsymbol{p}_2=\min_{\boldsymbol{p}_2}\max_{\boldsymbol{p}_1}
\boldsymbol{p}_1^T  \boldsymbol{O}\boldsymbol{p}_2={v}
\end{equation}
where ${\boldsymbol{p}_1}$ and ${\boldsymbol{p}_2}$ indicate the probability distributions over the related actions 
of player $1$ and player $2$, and ${v}$ is the value of the game. Moreover, $\boldsymbol{O}$ can be expressed as,
\begin{equation*}
\boldsymbol{O}=\big[\boldsymbol{o}_{1}^T, \boldsymbol{o}_{2}^T,
\ldots, \boldsymbol{o}_{l_1}^T\big]^ T
\end{equation*}
where ${\boldsymbol{o}_{i}}$ is ${1\times l_2}$ vector for $i$\!
$\in {(1,\ldots,l_1)}$. Hence, ${\boldsymbol{v}_1=\boldsymbol{p}_1^T
\boldsymbol{O}=}\sum_{i=1}^{N} p_{1,i}\boldsymbol{o}_{i}$ and
${v}\!\!=\boldsymbol{v}_1\boldsymbol{p}_2$. In addition, 
 we consider all the entries of matrix are more than zero. The value of the game, which contains $l_1$ actions with vectors
$\boldsymbol{o}_{1},\ldots,\boldsymbol{o}_{l_1}$, is denoted by
$\text{zerosum}(\boldsymbol{o}_{1},\dots,\boldsymbol{o}_{l_1})$.
First, we state a lemma in order to prove the proposition.\\
\begin{lemma} 
If the following relationship exists between $\boldsymbol{o}_{1},\ldots,\boldsymbol{o}_{l_1}$, player $1$ 
 can play the game without the $l_1$-th action while it gets the same value,
\begin{eqnarray}
 \boldsymbol{o}_{l_1}=\lambda_1  \boldsymbol{o}_{1}+\lambda_2  \boldsymbol{o}_{2}+\ldots+\lambda_{l_1-1}  \boldsymbol{o}_{l_1-1} \nonumber\\
\sum_{i=1}^{l_1-1} \lambda_i=1,~~
-\infty<\lambda_i<\infty, \forall i.\,\,\,
\end{eqnarray}
\end{lemma}
\begin{proof}
First, 
 assume that player $1$ 
has
 optimal probabilities $p_{1,1}^*,\ldots,p_{1, {l_1-1}}^*$ over $\boldsymbol{o}_{1}, \ldots, \boldsymbol{o}_{l_1-1}$, respectively. Equation (22) can be rewritten by the following representation,
\begin{eqnarray*}
 \boldsymbol{o}_{l_1}=\Bigg(\bigg(\Big(\big((h_{11}  \boldsymbol{o}_{1}+h_{12}  \boldsymbol{o}_{2} ) h_{21}+h_{22}
 \boldsymbol{o}_{3} \big) h_{31}+h_{32}
  \boldsymbol{o}_{4} \Big)+\ldots\bigg) h_{l_1-2  ,1}+h_{l_1-2  ,2}  \boldsymbol{o}_{l_1-1} \Bigg),
\end{eqnarray*}
 where the following relationships 
  exist between the set of $h_{k,1}$s and $h_{k,2}$s
\begin{eqnarray*}
\begin{matrix}
h_{k,1}+h_{k,2}=1, & 1\leq k\leq l_1-2.\\
\end{matrix}
\end{eqnarray*}
Moreover, 
 we have the next equations between
$\{h_{{k},1}, h_{{k},2}\}$ and $\{\lambda_{k}\}$ for $1\leq k\leq
l_1-2$,
\begin{eqnarray}
\begin{matrix}
h_{{l_1-2},2}=\lambda_{l_1-1}, &  h_{{l_1-2},1}=1-\lambda_{l_1-1},\\
h_{l_1-3,2}(h_{{l_1-2},1} )=\lambda_{l_1-2}, &
h_{{l_1-3},1}=1-h_{{l_1-3},2}, \\
\ldots & \dots \\
h_{2,2} \prod^{l_{1}-2}_{i=3} h_{i,1} =\lambda_{3} , & h_{2,1}
 = 1 - h_{2,2}\\
 h_{1,2}\prod^{l_{1}-2}_{i=2} h_{i,1} =\lambda_{2} , & h_{1,1} \prod^{l_{1}-2}_{i=2} h_{i,1} =\lambda_{1}.
\end{matrix}
\end{eqnarray}
Afterwards, we introduce a game containing $l_1$ actions with vectors 
$\boldsymbol{o}_{1},\ldots, \boldsymbol{o}_{{l_1}-1}$ and $
\boldsymbol{o}'_{l_1}=h_{1,1} \boldsymbol{o}_{1}+h_{1,2}
\boldsymbol{o}_{2}$. Besides, 
 the optimal probabilities of the new game are assumed as $q_{1,1}^*,\ldots,q_{1,l_1-1}^*$
and $q_{1{l_1}}^{'*}$. The value of game to which the $l_1$-th action is added is not less than the game without the $l_1$-th action according to \cite{Fudenberg}, meaning that,
\begin{eqnarray}
\text{zerosum}( \boldsymbol{o}_{1},\ldots,
\boldsymbol{o}_{l_1-1})\leq \text{zerosum}(
\boldsymbol{o}_{1},\ldots, \boldsymbol{o}_{l_1-1},
\boldsymbol{o}'_{l_1}=h_{11} \boldsymbol{o}_{1}+h_{12}
\boldsymbol{o}_{2}).
\end{eqnarray}
In other words
, for the new game, we have the following results,
\begin{eqnarray*}
\text{zerosum}( \boldsymbol{o}_{1},\ldots, \boldsymbol{o}_{l_1-1})
\leq \text{zerosum}( \boldsymbol{o}_{1},\ldots,
\boldsymbol{o}_{l_1-1}, \boldsymbol{o}'_{l_1}=h_{11}
\boldsymbol{o}_{1}+h_{12}  \boldsymbol{o}_{2} )
\\=\min_{v}(q_{11}^* \boldsymbol{o}_{1}+\ldots+ q_{1,l_1-1}^*  \boldsymbol{o}_{l_1-1}+ q^{'*}_{1l_1}  \boldsymbol{o}^{'}_{l_1})\qquad\;\;\;\;\qquad\qquad\qquad
\\=\min_{v}((q_{11}^* +h_{11} q^{'*}_{1 l_1})  \boldsymbol{o}_{1}+
(q_{12}^{*}+h_{12} q^{'*}_{1l_1})  \boldsymbol{o}_{2}  +\ldots
+q_{1,l_1-1}^*  \boldsymbol{o}_{l_1-1}),
\end{eqnarray*}
where $\min_{v}$ finds the entry with the minimum value of vector $v$. 
 If both $h_{1,1}$ and $h_{1,2}$ are not less than zero, set $(q_{11}^{*}+h_{1,1} , q_{12}^{*}+h_{1,2}, \dots , q_{1,{l_1-1}}^{*})$ can be interpreted as a distribution vector over $l_1 - 1$ actions of player $1$. Notice that each probability distribution over these selected actions brings the value not greater 
  than ${v}$. Thus, we can conclude that
\begin{eqnarray}
\text{zerosum}( \boldsymbol{o}_{1},\ldots,
\boldsymbol{o}_{l_1-1})\geq \text{zerosum}( \boldsymbol{o}_{1},
\ldots,  \boldsymbol{o}_{l_1-1},
 \boldsymbol{o}^{'}_{l_1}=h_{11}  \boldsymbol{o}_{1}+h_{12}  \boldsymbol{o}_{2}).
\end{eqnarray}
Due to (24) and (25), we have
\begin{eqnarray}
\text{zerosum}( \boldsymbol{o}_1,\ldots, \boldsymbol{o}_{l_1-1})=
\text{zerosum}( \boldsymbol{o}_{1}, \ldots,  \boldsymbol{o}_{l_1-1},
 \boldsymbol{o}^{'}_{l_1}=h_{11}  \boldsymbol{o}_{1}+h_{12}  \boldsymbol{o}_{2}).
\end{eqnarray}
In other words, if the action $l_1$ with 
vector $ \boldsymbol{o}^{'}_{l_1}=h_{1,1}
 \boldsymbol{o}_{1}+h_{1,2}  \boldsymbol{o}_{2}$ is eliminated, we will gain the same value. However,
if one of them is less 
 than zero, we cannot get the above formulation.
  Without loss of generality, we assume that $h_{1,1}<0$ and $-\alpha=q_{11}^* +h_{1,1} q^{'*}_{1{l_1}}<0$. Remind that $h_{1,1}+ h_{1,2}=1$, thus $h_{1,2}>0$ and therefore $q_{1,2}^* + h_{1,2} q^{'*}_{1,{l_1}} > 0$
. Because the summation
over probabilities is 1, hence,
\begin{eqnarray*}
\sum_{i=1}^{l_1-1} q^{*}_{1,i} + q^{'*}_{1,l_1} = 1\nonumber\qquad\qquad\qquad\qquad\qquad\\
(q_{1,1}^{*}+ h_{1,1} q^{'*}_{1{l_1}})+ (q_{1,2}^{*} +h_{1,2} q^{'*}_{1{l_1}})+\ldots +q_{1,({l_1}-1)}^{*}=1,\nonumber\qquad\\
(q_{1,2}^* +h_{1,2} q^{'*}_{1,{l_1}})+\dots + q_{1,({l_1}-1)}^{*}
=1+\alpha.\qquad\qquad\qquad
\end{eqnarray*}
Now, consider distribution vector $[T_2, T_3, \ldots, T_{l_1-1}]$ which is constructed by the following,
\begin{eqnarray}
\begin{cases}
\frac{(q_{1,2}^{*}+h_{1,2} q^{'*}_{1,{l_1}})}{(1+\alpha)}=T_2\\
\frac{q^{*}_{1,3}}{(1+\alpha)}=T_3\\
\ldots\\
\frac{q^{*}_{1,({l_1}-1)}}{(1+\alpha)}=T_{{l_1}-1}
\end{cases}
\end{eqnarray}
where $T_2+T_3+\dots+T_{{l_1}-1}=1$.
Again, we have the following inequality:
\begin{eqnarray}
\min_{v}( \boldsymbol{o}_{2} T_2+\ldots + \boldsymbol{o}_{{l_1}-1}
T_{{l_1}-1})\leq \text{zerosum}( \boldsymbol{o}_{2},\ldots,
\boldsymbol{o}_{{l_1}-1}) \leq \text{zerosum}( \boldsymbol{o}_{1},
\boldsymbol{o}_{2},\dots, \boldsymbol{o}_{{l_1}-1}).
\end{eqnarray}
 To put it differently, (28) can be reformulated as
\begin{eqnarray*}
\min_{v} \bigg(\frac{ \boldsymbol{o}_{2} p_{12}^*+\dots+
\boldsymbol{o}_{{l_1}-1} p_{1({l_1}-1)}^{*}}{1-p^{*}_{11}}\bigg)
\leq \text{zerosum}( \boldsymbol{o}_{2},\ldots,
\boldsymbol{o}_{{l_1}-1}) \leq \text{zerosum}( \boldsymbol{o}_{1},
\boldsymbol{o}_{2},\dots, \boldsymbol{o}_{{l_1}-1}).
\end{eqnarray*}
Besides, (24) gives that
\begin{eqnarray*}
\min_{v} \bigg((1+\alpha)( \boldsymbol{o}_{2} T_2+\dots+
\boldsymbol{o}_{{l_1}-1} T_{{l_1}-1})-\alpha
\boldsymbol{o}_{1}\bigg) \geqslant \text{zerosum}(
\boldsymbol{o}_{1}, \boldsymbol{o}_{2},\ldots,
\boldsymbol{o}_{{l_1}-1}),
\end{eqnarray*}
and,
\begin{eqnarray*}
\min_{v} \Bigg(\bigg(\frac{ \boldsymbol{o}_{2} p_{12}^{*}+\dots+
\boldsymbol{o}_{{l_1}-1}
p_{1({l_1}-1)}^{*}}{1-p^{*}_{11}}\bigg)(1-p_{11}^{*})+p_{11}^{*}
 \boldsymbol{o}_{1}\Bigg) \geqslant \text{zerosum}( \boldsymbol{o}_{1},  \boldsymbol{o}_{2},\ldots, \boldsymbol{o}_{{l_1}-1}).
\end{eqnarray*}
If $b_k$, $c_k$ and $d_k$ are the $k$-th the entries of $( \boldsymbol{o}_{2} T_2+\dots+  \boldsymbol{o}_{l_1-1} T_{{l_1}-1})$,
$\bigg(\frac{ \boldsymbol{o}_{2} p_{12}^{*}+\dots+
\boldsymbol{o}_{{l_1}-1} p_{1({l_1}-1)}^{*}}{1-p^{*}_{11}}\bigg)$
and $ \boldsymbol{o}_{1}$, respectively, we could obtain the
following result,
\begin{eqnarray}
b_k>\frac{d_k}{\alpha+1}+\frac{\text{zerosum}( \boldsymbol{o}_{1}, \boldsymbol{o}_{2},\dots, \boldsymbol{o}_{{l_1}-1})}{\alpha+1}\geqslant \qquad\qquad\qquad\qquad\qquad\qquad\qquad\qquad\nonumber\\
\frac{\alpha(\text{zerosum}( \boldsymbol{o}_{1},
 \boldsymbol{o}_{2},\ldots, \boldsymbol{o}_{{l_1}-1})-p_{11}^{*} c_k)} {(\alpha+1)(1-p_{11}^{*})}
+\frac{\alpha(\text{zerosum}( \boldsymbol{o}_{1},
 \boldsymbol{o}_{2},\ldots, \boldsymbol{o}_{{l_1}-1}))} {\alpha+1}=
\nonumber\\\text{zerosum}( \boldsymbol{o}_{1},
\boldsymbol{o}_{2},\dots,
\boldsymbol{o}_{{l_1}-1})\bigg(\frac{\alpha}{(1+\alpha)(1-p_{11}^{*})}+\frac{1}{\alpha+1}\bigg)
-\frac{\alpha p_{11}^{*}) c_k}{(1+\alpha)(1-p_{11}^{*})}.
\end{eqnarray}
Consequently, 
\begin{eqnarray}
\text{zerosum}( \boldsymbol{o}_{1},\dots,
\boldsymbol{o}_{{l_1}-1})\bigg(\frac{\alpha}{(1+\alpha)(1-p_{11}^{*})}+\frac{1}{\alpha+1}\bigg)
< \frac{\alpha p_{11}^{*} c_k}{(1+\alpha)(1-p_{11}^{*})}.
\end{eqnarray}
This suggests the following inequality for $b_k$, $c_k$ and $d_k$
\begin{eqnarray}
w_1 c_k+w_2 b_k>d_k    \qquad st. \qquad w_1+w_2=1.
\end{eqnarray}
Therefore, consideration of both (30) and (31) gives us the following inequality,
\begin{eqnarray}
\text{zerosum}\big( \boldsymbol{o}_{1},  \boldsymbol{o}_{2}, \dots,  \boldsymbol{o}_{{l_1}-1}\big)<\nonumber\qquad\qquad\qquad\qquad\qquad\qquad\qquad\qquad\qquad\qquad\nonumber\\
 \min_{v} \Bigg(w_1 \bigg( \frac{ \boldsymbol{o}_{2} p_{12}^{*}+\dots + \boldsymbol{o}_{{l_1}-1} p_{1,{l_1}-1}^{*}}{1-p^{*}_{11}}\bigg)+
 w_2 ( \boldsymbol{o}_{2} T_2+\dots + \boldsymbol{o}_{{l_1}-1} T_{{l_1}-1})\Bigg)\nonumber\\
 = \min_{v} (\beta_2  \boldsymbol{o}_{2} + \dots+ \beta_{{l_1}-1}
 \boldsymbol{o}_{{l_1}-1})\qquad\qquad\qquad\qquad\qquad\qquad\qquad\qquad\,\,\,\,
\end{eqnarray}
in which $\beta_2+ \beta_3+\dots+\beta_{l_1-1}=1$ , $\beta_2$,
$\beta_3$ ,$\dots,\beta_{{l_1}-1}\geq 0$. We know that the 
minimum entry of vector $\beta_2  \boldsymbol{o}_{2} + \dots+ \beta_{{l_1}-1}
 \boldsymbol{o}_{{l_1}-1}$ is not higher than $\text{zerosum} ( \boldsymbol{o}_{1}, \boldsymbol{o}_{2},\dots, \boldsymbol{o}_{{l_1}-1})$ for any set of $\beta_2$,
$\beta_3$ ,$\dots,\beta_{{l_1}-1}$.
Therefore, our initial assumption is not
correct. In other words, $q_{1,1}^{*}+h_{1,1} q^{'*} h_{1,{l_1}}$ is
not less than zero, and we can obviate $
\boldsymbol{o}_{l_1}^{'}$.
 It means that
\begin{eqnarray}
\text{zerosum}( \boldsymbol{o}_{1},\dots, \boldsymbol{o}_{{l_1}-1})=
\text{zerosum}( \boldsymbol{o}_{1},\dots, \boldsymbol{o}_{{l_1}-1},
 \boldsymbol{o}^{'}_{l_1}=h_{1,1}  \boldsymbol{o}_{1}+h_{1,2}  \boldsymbol{o}_{2}).
\end{eqnarray}
Returning to the general case in (22), it can be concluded from (33) that
\begin{eqnarray*}
\text{zerosum}(\boldsymbol{o}_{1},\dots, \boldsymbol{o}_{{l_1}-1})=
\text{zerosum}( \boldsymbol{o}_{1},\dots, \boldsymbol{o}_{{l_1}-1},
\boldsymbol{o}^{'}_{l_1}=h_{1,1}  \boldsymbol{o}_{1} + h_{1,2}
\boldsymbol{o}_{2})=\qquad\,\,\,\qquad\qquad\,\,\,\,
\\\text{zerosum}(\boldsymbol{o}_{1}, \dots, \boldsymbol{o}_{{l_1}-1}, \boldsymbol{o}^{'}_{l_1}, h_{2,2} \boldsymbol{o}_{3} + h_{2,1} \boldsymbol{o}^{'}_{l_1} )
= \text{zerosum}( \boldsymbol{o}_{1},\dots,
\boldsymbol{o}_{{l_1}-1},h_{2,2}
\boldsymbol{o}_{3}+h_{2,1}\boldsymbol{o}^{'}_{l_1} )=\,\,\,\,\,
\\ \text{zerosum}\bigg( \boldsymbol{o}_{1},\dots, \boldsymbol{o}_{{l_1}-1},h_{22}  \boldsymbol{o}_{3}+h_{21}  \boldsymbol{o}^{'}_{l_1},
h_{3,1}(h_{2,2}  \boldsymbol{o}_{3}+ h_{2,1}  \boldsymbol{o}^{'}_{l_1}
)+ h_{3,2} \boldsymbol{o}_{4}\bigg)= \ldots =\qquad\qquad\,\,\,\,\,
\\\text{zerosum}\bigg( \boldsymbol{o}_{1},\dots, \boldsymbol{o}_{{l_1}-1}, h_{31}(h_{22}  \boldsymbol{o}_{3}+ h_{21}  \boldsymbol{o}^{'}_{l_1} )+ h_{32}  \boldsymbol{o}_{4}\bigg)=\,\,\,\,\,\,\qquad\qquad\qquad
\qquad\qquad\,\,\,\,\,\,\,\,\,\,\,\,\,\,\,\,\;\;\;
\\\text{zerosum}\Bigg(\bigg(\Big(\big((h_{11}
\boldsymbol{o}_{1}+h_{12} \boldsymbol{o}_{2}) h_{21} + h_{22}
\boldsymbol{o}_{3}\big) h_{31}
 +h_{32}  \boldsymbol{o}_{4} \Big)+\dots\bigg) h_{({l_1}-2)1}+ h_{({l_1}-2)2}
\boldsymbol{o}_{{l_1}-1}\Bigg)=\,
\end{eqnarray*}
\begin{eqnarray}
\text{zerosum} ( \boldsymbol{o}_{1},\dots, \boldsymbol{o}_{{l_1}-1},
\boldsymbol{o}_{l_1}).\;\;\;\,\,\,
\end{eqnarray}
\end{proof}
The expression in (34) states that if $
\boldsymbol{o}_{l_1}$ satisfies the conditions of \emph{Lemma} 1, we
can omit it.  Here, we return to prove \textit{Proposition $1$}.
Each $ \boldsymbol{o}_{i}$ has $l_2$ entries, so we can represent
all of the $l_1$ vectors by at most $l_2$ vectors of them. These
basic vectors are linear and independent. Without loss of
generality, we assume that these vectors are $
\boldsymbol{o}_{1},\dots, \boldsymbol{o}_{{l^{'}_2}},$ where $
{l^{'}_2}\leq {l_2}$. Based upon
the vector representation, all 
 $\{\boldsymbol{o}_{i}\}$s are classified into three groups.

Now, we assume that each $ \boldsymbol{o}_{i}$ can be displayed by
$\boldsymbol{o}_{i}=\lambda_{1,i}  \boldsymbol{o}_{1}+\dots
+\lambda_{{l^{'}_2},i}  \boldsymbol{o}_{{l^{'}_2}}$. Also, the
coefficients are unique due to 
 linearity and independency. These 
 groups are stated as follows:

\textbf{Group I: } If $\sum_{j=1}^{{l^{'}_2}} \lambda_{j,i}=1$, we
can obviate $ \boldsymbol{o}_{i}$ and get the same value as in
\emph{Lemma} 1.

\textbf{Group II:} If $\sum_{j=1}^{{l^{'}_2}} \lambda_{j,i}< 1$, we have
the following facts:

We assume that the optimal probability distributions over the $l_1$ actions
are $p_{1,1}^{*},\dots,p_{1,{l_1}}^{*}$.
\begin{eqnarray*}
\text{zerosum}( \boldsymbol{o}_{1},\dots, \boldsymbol{o}_{l_1})\geq
 \text{zerosum}( \boldsymbol{o}_{1},\dots, \boldsymbol{o}_{i-1},  \boldsymbol{o}_{i+1}, \dots , \boldsymbol{o}_{l_1})
\end{eqnarray*}
where the second term does not include $ \boldsymbol{o}_{i}$. From~\cite{non_game}, we
know that
\begin{eqnarray}
\text{zerosum}( \boldsymbol{o}_{1},\dots, \boldsymbol{o}_{{l_1}}) =
\min_{v}(p_{1,1}^{*}  \boldsymbol{o}_{1}+\dots+p_{1,{l_1}}^{*}
 \boldsymbol{o}_{l_1}).
\end{eqnarray}
Now, we extend (35) as
\begin{eqnarray*}
p_{1,1}^{*}  \boldsymbol{o}_{1}+\dots+p_{1,i}^{*}
 \boldsymbol{o}_{i}+\dots+p_{1,{l_1}}^{*}
 \boldsymbol{o}_{l_1}= p_{1,1}^{*}  \boldsymbol{o}_{1}+ \dots +p_{1,i}^{*} (\lambda_{1i}  \boldsymbol{o}_{1}+ \dots +\lambda_{l'_2 i}  \boldsymbol{o}_{l'_2})+
\dots +p_{1,{l_1}}^{*}  \boldsymbol{o}_{{l_1}}
\\< p_{1,1}^{*}  \boldsymbol{o}_{1} +\dots + p_{1,i}^{*} (\lambda_{1,i}  \boldsymbol{o}_{1}+ \dots+\lambda_{l'_2 i}  \boldsymbol{o}_{l'_2} )+\dots + p_{1,{l_1}}^{*}  \boldsymbol{o}_{l_1}
+\Bigg(1-\sum_{j=1}^{l'_2-1} \lambda_{ji} -\lambda_{l'_2 i}\Bigg)
\boldsymbol{o}_{l'_2}
\end{eqnarray*}
\begin{eqnarray}
\nonumber=p_{1,1}^{*}  \boldsymbol{o}_{1}+\dots +p_{1,i}^{*}
\bigg(\lambda_{1,i}  \boldsymbol{o}_{1}
+\Big(1-\sum_{j=1}^{l'_2-1}\lambda_{i,j}\Big)  \boldsymbol{o}_{l'_2}
\bigg)+\ldots+p_{1,{l_1}}^{*}  \boldsymbol{o}_{{l_1}}\,\qquad\,\:\,
\end{eqnarray}
in which we call the expression stated in parenthesis in the last
term as $ \boldsymbol{o}'_i$. The value of the game when playing
with $ \boldsymbol{o}'_i$ instead of $ \boldsymbol{o}_{i}$ is given
via
\begin{eqnarray}
\nonumber \text{zerosum}( \boldsymbol{o}_{1},\dots ,
\boldsymbol{o}^{'}_i,\dots,  \boldsymbol{o}_{l_1})\geq
\min(p_{11}^{*}  \boldsymbol{o}_{1}+ \dots+p_{1i}^{*}
\boldsymbol{o}^{'}_i+\dots+p_{1{l_1}}^{*}
 \boldsymbol{o}_{l_1})>\nonumber\\\min(p_{11}^{*}
 \boldsymbol{o}_{1}+\dots+p_{1{l_1}}^{*}
 \boldsymbol{o}_{l_1})=\text{zerosum}( \boldsymbol{o}_{1},\dots, \boldsymbol{o}_{l_1}).\;\;\;\qquad\nonumber
\end{eqnarray}
Now, 
 $ \boldsymbol{o}_{i}$ can be represented by basic vectors $ \{\boldsymbol{o}_{1},\dots ,
\boldsymbol{o}^{'}_i,\dots, \boldsymbol{o}_{l_1}\}$ in which the sum of coefficients becomes 
 $1$. Thus, we can obviate $ \boldsymbol{o}'_i$ and at the same time get the same
value. In other words, we have
\begin{eqnarray*}
\text{zerosum}( \boldsymbol{o}_{1},\dots, \boldsymbol{o}_{i-1},
\boldsymbol{o}^{'}_i,  \boldsymbol{o}_{i+1},\dots,
\boldsymbol{o}_{l_1})= \text{zerosum}( \boldsymbol{o}_{1},\dots,
\boldsymbol{o}_{i-1}, \boldsymbol{o}_{i+1},\dots,
\boldsymbol{o}_{l_1}).
\end{eqnarray*}
Moreover, we have
\begin{eqnarray*}
\text{zerosum}( \boldsymbol{o}_{1},\dots, \boldsymbol{o}_{l_1} )\geq
 \text{zerosum}( \boldsymbol{o}_{1},\dots, \boldsymbol{o}_{i-1}, \boldsymbol{o}_{i+1},\dots, \boldsymbol{o}_{l_1}).
\end{eqnarray*}
Then, we can remove all of vectors which have coefficients
satisfying the following inequality without loss in the value of the
game, $v$.
\begin{eqnarray}
\sum_{j=1}^{l'_2} \lambda_{ji}<1.
\end{eqnarray}
\textbf{Group III:}     If $\sum_{j=1}^{l'_2} \lambda_{ji}> 1$, we can
show $ \boldsymbol{o}_{i}$ by the following equation
\begin{eqnarray}
 \boldsymbol{o}_{i}=\lambda_{1,i}  \boldsymbol{o}_{1}+ \lambda_{2,i}  \boldsymbol{o}_{2}+\dots+\lambda_{{l_2}^{'},i}  \boldsymbol{o}_{l'_2-1}.
\end{eqnarray}
In this case, there exists at least one coefficient, e.g.,
$\lambda_{l'_2 , i}$, which is greater 
 than zero. Now, we try to show $\boldsymbol{o}_{l'_2}$ by $ \boldsymbol{o}_{1},\dots,
\boldsymbol{o}_{l'_2-1}$ including $ \boldsymbol{o}_{i}$. Indeed,
\begin{eqnarray*}
 \boldsymbol{o}_{l'_2}=\frac{-(\lambda_{1i}  \boldsymbol{o}_{1}+\lambda_{2i}
 \boldsymbol{o}_{2}+\dots+\lambda_{l'_2-1 i}  \boldsymbol{o}_{l'_2-1})}{\lambda_{l'_2 i}}
+\frac{1}{\lambda_{l'_2 i}}   \boldsymbol{o}_{i} =\mu_1
\boldsymbol{o}_{1}+\mu_2  \boldsymbol{o}_{2}+\dots+\mu_{l'_2-1}
 \boldsymbol{o}_{l'_2-1}.
\end{eqnarray*}
However, we know that
\begin{eqnarray}
\mu_1+\dots+\mu_{l'_2-1}=
\frac{-(\lambda_{1i}+\lambda_{2i}+\dots+\lambda_{{l_2}^{'}-1, i})+1
}{\lambda_{l'_2 i}}> 1.
\end{eqnarray}
Therefore, we can remove $ \boldsymbol{o}_{{l_2}^{'}}$ according to
the second group. As a result, 
 we only need the $l_2$ actions among the $l_1$ ones
  and get the same value. Similar classification can be applied to vectors of $
\boldsymbol{o}_{1},\dots, \boldsymbol{o}_{l_1}$ and at each stage
one vector is removed. Finally, $l'_2(l'_2\leqslant l_2)$ actions
(vectors) remain for playing the game.


%
\vspace{.3 cm}


\section{Proof of Proposition 2}
First, we consider an assumption for each allocation to continue our
proof.
\newtheorem{Ass}{Assumption}
\begin{Ass}
If the following relation exist between vectors of allocations,
$\{\boldsymbol{h}_1, \boldsymbol{h}_2,\ldots, \boldsymbol{h}_{M'},\\
\boldsymbol{h}_{M'+1}\}$, \vspace{-1 cm}
\begin{eqnarray}
\sum_{i=1}^{M'}\lambda_i\boldsymbol{h}_i=
\boldsymbol{h}_{M'+1},\vspace{-.6 cm}
\sum_{i=1}^{M'}\lambda_{i}=1,\qquad\qquad\qquad\,\,\,\\\vspace{-.4
cm}\ \text{for } {1 \leq i \leq M'+1}, 
EL\{\boldsymbol{h}_i\}\nsubseteq
EL\{\boldsymbol{h}_1,\dots,\boldsymbol{h}_{i-1},\boldsymbol{h}_{i+1},\dots,\boldsymbol{h}_{M'+1}\},\vspace{-.4
cm}\nonumber
\end{eqnarray}
then the occurrence probability of all relations is zero.
\end{Ass}

Indeed, each $\boldsymbol{h}_i$ has $M'$ entries. Accordingly, the
allocation vectors construct $M'$ equations, and we have $M'-1$
parameters involving $\lambda_1,\ldots,\lambda_{M'-1}$. For (39), we have $M'-1$ parameters satisfying $M'$ equations. 
This
situation makes hard to yield these $M'-1$ parameters out of $M'$
equations. For instance, if each $a_{ij,k}$ is independent with
respect to the other $a_{i'j',k'}$ with the uniform distribution,
this assumption is precise. Also, our simulation result certifies
the assumption.
 We assume that the attacker's strategy is the same as the strategy of
the original zero-sum game. In the original form, we have the 
 following equation \vspace{-.1 cm}
\begin{eqnarray*}
\boldsymbol{p}_1^{*T}\boldsymbol{U}\boldsymbol{p}_2^{*}=\max_{\boldsymbol{p}_1} \min_{\boldsymbol{p}_2} \boldsymbol{p}_1^{T}\boldsymbol{U}\boldsymbol{p}_2,\nonumber\\
\boldsymbol{U}\boldsymbol{p}_2^* =
\left(\begin{IEEEeqnarraybox*}[][c]{,c/c/c,}
{\boldsymbol{u}}_1\\
\ldots\\
{\boldsymbol{u}}_{\frac{N!}{(N-M')!}}
\end{IEEEeqnarraybox*}\right),\vspace{-.5 cm}
\end{eqnarray*}
in which $\boldsymbol{p^*}_1$ and $\boldsymbol{p^*}_2$ are the
optimal action probabilities of the coordinator and attacker,
respectively, and $\boldsymbol{U}$ is the payoff matrix in
accordance with the zero-sum game between the coordinator and the
attacker.

The related ${\boldsymbol{u}}_i$'s to the allocations with the
non-zero probabilities, which are named as the proper allocations, are
the same as the overall value of the game and
$\max({\boldsymbol{u}}_1,\ldots,{\boldsymbol{u}}_{\frac{N!}{(N-M')!}})$.
According to \textit{Proposition $1$}, we only need $M'$ proper
allocations, namely complete allocations, to obtain the similar value
when using all actions. Hence, we must show that each complete
allocation is surely selected as the solution of the PC-game at
least one action by means of contradiction. Notice that if more than
one proper allocations exist in the first-auction, only one of them
is randomly selected as the solution of auction. Furthermore, we
know that each allocation of channels can be found $M'^{(N-M')}$
times at the first auction. The worst case occurs for a complete
allocation, for example $\boldsymbol{J}$, when at least one of the
other proper allocations always exists in the first auction
including this allocation. For simplicity, we assume that
$\boldsymbol{J}$ is $(1,2,\ldots,M')$. Now, consider the following
first auctions: \vspace{-.2 cm}
\begin{eqnarray}
\left\{\begin{array}{ll} {1:} &\textrm{$(1,2,\ldots,M',1,\ldots,1)$}\\
{2:} &\textrm{$(1,2,\ldots,M',2,\ldots,2)$}\\
\ldots\\
{M':} &\textrm{$(1,2,\ldots,M',M',\ldots,M')$}\\
\end{array}\right.\vspace{-.8 cm}
\end{eqnarray}
where $(1,2,\ldots,M',j,\ldots,j)$ means that this first-auction
includes allocation $\boldsymbol{J}$, and the coordinator selects
channel $j$ for the remaining users as well. Hence, we have at least
$M'+1$ proper allocations among the above actions. These vectors
cannot be linearly independent since dimension of vectors is $M'$.
Therefore, we have the following according to
\textit{Proposition 1}.
\begin{eqnarray}
\boldsymbol{J}=\lambda_1 \boldsymbol{h}_{
\boldsymbol{o}_{1}}+\lambda_2 \boldsymbol{h}_{
\boldsymbol{o}_{2}}+\ldots+\lambda_{M'} \boldsymbol{h}_{
\boldsymbol{o}_{M'}} \qquad \lambda_1+\ldots+\lambda_{M'}=1,
\end{eqnarray}
where $\boldsymbol{h}_{ \boldsymbol{o}_{i}}$s are the proper
allocations. Also, any 
 two allocations 
  of $M'+1$ allocations
differ from two elements so that the conditions of
\textit{Assumption 1} are satisfied, and the probability of this occurrence is zero. Hence, our initial assumption about the concurrent existence of these allocations, is not correct and the PC-game is equal to the original game.
\vspace{-.25 cm}
\section{Proof of Proposition 3}
In the PD-game, payment for user $i$ has two parts, $p_{i,1}$ and
$p_{i,2}$, which are related to first and second auctions. If we
assume that the SUs choose preferences $l_1,\ldots,l_N$ as their
actions and channel $j$ is dedicated to the $i$-th SU while the
attacker jams channel $h$, we have the following formulation for the
average profit of SU $i$. \vspace{.5 cm} \vspace{-.7cm}
\begin{eqnarray}
\mathbb{E}\big(v_i-p_{i1}^t(\widehat{v}_i,v_{-i})-p_{i2}^t(\widehat{v}_i,v_{-i})\mid
l_1,\ldots,l_N
\big)\qquad\qquad\qquad\qquad\qquad\qquad\qquad\qquad\qquad\qquad \nonumber\\
=a_{ij,h}Q_{2,h}(t)+\sum_{h=1}^{M'}\Bigg(\Big[\sum_{k=1, k\neq i}^N
 z_{k l_{k}}^{t,opt} {{a'}_{k l_{k}}}(t)
-\max_{Z \mid a_{i j'}=0 \forall j'}\sum_{k=1, k\neq i}^N
 z_{k l_{k}}^{t} {a'}_{k l_{k}}(t)\Big]\Bigg)Q_{2,h}(t)\qquad\qquad\nonumber\\
+\sum_{h=1}^{M'}\Big[\sum_{k=1, k\neq i, k \in S_1}^N \sum_{j'=1,j'
\in S_2}^{M'} z_{k j'}^{t,opt} {a'}_{k j'}(t)- \max_{Z \mid a_{i
j'}=0 \forall j' }\sum_{k=1, k\neq i,k \in S_1}^N \sum_{j'=1,j' \in
S_2}^{M'} z_{k j'}^{t} {a'}_{k j'}(t)\Big]Q_{2,h}(t)).
\end{eqnarray}\vspace{0 cm}
In (42), $v_i$ and $\widehat{v}_i$  are the actual and submitted bid for SU $i$. Moreover, $Q_{2,h}$, $S_1$ and $S_2$ are
the probability for jamming of channel $h$ by the jammer, the set of SUs and
channels remained from first auction, respectively. To attain the
$i$-th user's profit, we should apply the probability of preferences
for all the SUs. Moreover, $\boldsymbol{p}_2(t)\approxeq Q_2(t)$,
therefore, $\sum_h Q_{2,h}(t)a_{ij,h}={{a'}_{ij}}$ similar to $(12)$.
Hence, the expectation profit can be stated as follows,
\begin{eqnarray}
\mathbb{E}\big(v_i-p_i^t(v_i,v_{-i})\big)=\qquad\qquad\qquad\qquad\qquad\qquad\qquad\qquad\qquad\qquad
\qquad\qquad\qquad\qquad\nonumber\\
\sum_{o_1=1}^{M'} Q_{l_1} \sum_{o_2=1}^{M'} \ldots
\sum_{o_N=1}^{M'}Q_{l_2} \ldots Q_{l_N} \Big[\sum_{k=1}^N z_{k
l_{k}}^{t,opt} {{a'}_{k l_{k}}}(t) -\max_{Z \mid a_{ij'}=0 \forall
j} \sum_{k=1, k\neq i}^N  z_{k l_{k}}^{t}
{a'}_{k l_{k}}(t)\Big] +\nonumber\\
\Bigg[\sum_{\substack{k=1, k\neq i,\\ k \in S_1}}^N
\sum_{\substack{j=1, j\neq j'\\j \in S_2}}^{M'} z_{kj}^{t,opt}
{a'}_{kj}(t)-\max_{Z \mid a_{ij}=0 \forall j}\sum_{\substack{k=1\\k
\in S_1}}^N \sum_{\substack{j=1,\\j \in S_2}}^{M'} z_{kj}^{t}
{a'}_{kj}(t)\Bigg].\qquad\qquad\qquad\qquad\qquad
\end{eqnarray}
Equivalently,
\begin{eqnarray}
\mathbb{E}\big(v_i-p_i^t(v_i,v_{-i})\big)=\qquad\qquad\qquad\qquad\qquad\qquad\qquad\qquad\qquad\qquad
\qquad\qquad\qquad\qquad\nonumber\\
\sum_{o_1=1}^{M'} Q_{l_1} \sum_{o_2=1}^{M'} \ldots \sum_{o_N=1}^{M'}
Q_{l_2} \ldots Q_{l_N} \Big[\underbrace{ \sum_{k=1}^N z_{k
l_{k}}^{t,opt} {{a'}_{k
l_{k}}}(t)}_{1}+\underbrace{\sum_{\substack{k=1, k\neq i\\ k \in
S_1}}^N \sum_{\substack{j=1, j\neq j'\\j \in S_2}}^{M'}
z_{kj}^{t,opt} {a'}_{kj}(t)}_{2}\Big]\nonumber\\
 - \Big[\underbrace{\max_{Z
\mid a_{ij'}=0 \forall j} \sum_{\substack{k=1\\ k\neq i}}^N z_{k
l_{k}}^{t} {a'}_{k l_{k}}(t)}_{3}+ \underbrace{\max_{Z \mid a_{ij}=0
\forall j}\sum_{\substack{k=1\\k \in S_1}}^N \sum_{\substack{j=1\\j
\in S_2}}^{M'} z_{kj}^{t}
{a'}_{kj}(t)}_{4}\Big].\qquad\qquad\qquad\qquad\,\,\,
\end{eqnarray}
 The third and fourth terms are not function of the $i$-th SU. Therefore, we can disregard them to further analysis. But, the summation over the first and second terms are equal to the total profit of the 
 SUs according to $(15)$. In other words, the individual profit is equivalent to the total profit. For this reason, the rational SUs must bid truthfully.
\vspace{-.5 cm}

%
%

\ifCLASSOPTIONcaptionsoff
  \newpage
\fi

\end{document}